\documentclass[journal,twoside,web]{ieeeconf} 
\usepackage{url,hyperref}
\usepackage{graphicx}
\usepackage{amsmath,amssymb,amsfonts}
\usepackage{breqn}
\usepackage{tensor}
\usepackage{braket}
\usepackage{subcaption}
\usepackage{array}
\usepackage{graphicx}
\usepackage{multirow}
\usepackage{multicol}
\usepackage{color,colortbl}
\usepackage{soul}
\usepackage{comment}
\usepackage{float}
\usepackage{cancel}
\usepackage{siunitx}
\usepackage{flushend}
\usepackage{todonotes}
\allowdisplaybreaks
\usepackage[switch, mathlines]{lineno}
\newtheorem{theorem}{Theorem}

\newtheorem{lemma}{Lemma}

\newtheorem{remark}{Remark}

\DeclareMathOperator{\tr}{\mathrm{Tr}}

\DeclareMathOperator{\Rs}{\mathsf{R}_{\mathcal{S}} }
\DeclareMathOperator{\R}{ \mathsf{R} }
\DeclareMathOperator{\Kn}{ \mathsf{K}_n }
\DeclareMathOperator{\X}{ \mathsf{X} }

\DeclareMathOperator{\M}{\mathsf{M}}
\DeclareMathOperator{\Qn}{\mathsf{Q}_n}
\DeclareMathOperator{\Zn}{ \mathsf{Z}_n }

\DeclareMathOperator{\sinc}{\mathrm{sinc}}

\DeclareMathOperator{\Sn}{\mathsf{S}_n}

\DeclareMathOperator{\W}{\mathsf{W}}

\DeclareMathOperator{\rr}{\mathsf{r}}
\DeclareMathOperator{\r0}{\mathsf{r}_0}
\DeclareMathOperator{\rf}{\mathsf{r}_f}
\DeclareMathOperator{\A}{\mathsf{A}}

\title{Geometric Interpretation of Sensitivity to Structured Uncertainties in Spintronic Networks}
\author{S.\,P.\ O'Neil$^{1,*}$, \and
E.\,A.\ Jonckheere$^2$, \and S.\ Schirmer$^3$
\thanks{Any opinions in this work are solely that of the authors and do not reflect those of the U.S. Army, the USMA, or the DoD.}
\thanks{$^1$ Department of Electrical Engineering \& Computer Science, United States Military Academy, NY, USA. {\tt sean.oneil@westpoint.edu}}
\thanks{$^2$ Dept of Electrical \& Computer Engineering, University of Southern California, CA, USA. 
{\tt jonckhee@usc.edu}}
\thanks{$^3$ Faculty of Science \& Engineering, Physics, Swansea University, UK. {\tt s.m.shermer@gmail.com}}
}
\begin{document}
\maketitle
\thispagestyle{empty}

\begin{abstract}
We present a geometric model of the differential sensitivity of the fidelity error for state transfer in a spintronic network based on the relationship between a set of matrix operators. We show an explicit dependence of the differential sensitivity on the fidelity (error), and we further demonstrate that this dependence does not require a trade-off between the fidelity and sensitivity. We prove that for closed systems, ideal performance in the sense of perfect state transfer is both necessary and sufficient for optimal robustness in terms of vanishing sensitivity. We demonstrate the utility of this geometric interpretation of the sensitivity by applying the model to explain the sensitivity versus fidelity error data in two examples.
\end{abstract}
\section{Introduction}
Exploiting the unique properties of quantum systems has the potential to revolutionize applications in communications, sensing, and computing~\cite{koch_2022_quantum_optimal_control_survey,Glaser_2015}. However, harnessing any ``quantum advantage'' requires controls that not only meet exacting performance criteria, but retain a high level of performance in the face of model uncertainty.  While classical $H_\infty$ methods with guaranteed performance margins are applicable to a subset of open quantum systems undergoing weak continuous measurement~\cite{wang_2023,James_2007}, such methods, requiring nominal and robust stability of the controlled system, are ill-suited to the marginally stable dynamics of closed quantum systems. As such, post-synthesis selection of acceptable controllers based on an assessment of the robustness or the sensitivity of the performance measure is common. Various methods for delivering such assessments can be found in the literature, ranging from stochastic~\cite{khalid_2023_rim,koswara_2021} to analytic~\cite{Kosut_2013,jonckheere_2018_jt_test,o'neil_2024_log_sens}. 
\par Previous work has applied statistical analysis to a large set of controllers to determine whether a trade-off between performance and robustness is a fundamental limitation in quantum control problems. Such a trade-off might be expected from the fundamental limitation of classical feedback control, encapsulated by the frequency domain identity $S+T=I$, where $S$ and $T$ represent the error and log-sensitivity, respectively.  In~\cite{o'neil_2024_log_sens}, we demonstrated a trade-off between performance and robustness in the time domain in the sense that the normalized log-sensitivity of the fidelity error diverges as the error approaches zero. However,~\cite{jonckheere_2018_jt_test} and~\cite{o'neil_2023_data_set_1} suggest that such a trade-off does not necessarily hold for suboptimal controllers, those controllers that yield a non-zero fidelity error, in that the best-performing controllers (lowest fidelity error) are not necessarily those with the worst robustness (magnitude of the log-sensitivity). Further, \cite{Jonckheere_2017} provided examples of state transfer in spin networks where controllers with the best fidelity exhibited the smallest unnormalized differential sensitivity to parameter variations. 
\par Here we address the determination of an \textit{analytic} formulation relating fidelity to sensitivity which facilitates provable deductions on possible fundamental limitations without recourse to statistical analysis. To accomplish this, we develop a geometric model of the sensitivity of the fidelity error to parametric uncertainty that yields an explicit relationship between sensitivity and transfer fidelity. The remainder of the paper is organized as follows. In Section~\ref{sec: prelim} we review state transfer dynamics in spin networks and the derivation of the sensitivity of the fidelity error. In Section~\ref{sec: geometry_of_sensitivity} we develop the geometric model of the differential sensitivity. We follow with a major deduction of the model, the statement and proof that the sensitivity vanishes if and only if perfect state transfer is achieved. In Section~\ref{sec: examples} we employ the geometric model to analyze observed sensitivity-fidelity data in two state transfer examples. We conclude in Section~\ref{sec: conclusion}. 
\section{Preliminaries}\label{sec: prelim}
\subsection{Physical Model and Control Problem}
As in~\cite{o'neil_2023_data_set_1}, we consider a coupled network of $N$ spin-$1/2$ particles with $2^N \times 2^N$ real Hamiltonian 
\begin{equation}
   H_{tot} := \sum_{m\neq n} J_{nm} \left(X_n X_m + Y_n Y_m + \kappa Z_n Z_m\right). \label{eq:xxcoupling_system_hamiltonian} 
\end{equation}
$J_{mn}$ is the coupling strength between spin $m$ and $n$ and $\set{X_n,Y_n,Z_n}$ is the respective Pauli operator acting on spin $n$, formally the $N$-fold tensor product of $(n-1)$ copies of the $2 \times 2$ identity matrix with a Pauli matrix $\set{\sigma_x,\sigma_y,\sigma_z}$ in the $n$th position.  We consider the case of uniform coupling with $J_{mn} = J$ for all pairs $(m,n)$.  We restrict the dynamics to the single excitation subspace (SES), the subspace of the Hilbert space isomorphic to $\mathbb{C}^{2^N}$ that corresponds to a single excited spin in the network. This is justified when the goal is to transfer the state of the system from that of a single excited ``input'' spin to a single excited ``output'' spin with all other spins in the ground state, tantamount to transfer of a single qubit of information.  While control of quantum systems in an open-loop manner via optimally shaped time-varying fields is common in quantum control~\cite{koch_2022_quantum_optimal_control_survey}, we consider the alternative paradigm of shaping the energy landscape to facilitate the desired evolution~\cite{donovan_2011_hamitonian_manipulation,zhang_2023_gate_field_control}. In the interest of model simplicity, we introduce time-invariant external fields to maximize the probability of state transfer from a given pure input state $\ket{\psi_0}$ to a pure output state $\ket{\psi_f}$. As detailed in~\cite{Schirmer_2017}, these control fields enter the Hamiltonian as scalars $\Delta_n$ on the diagonal of the SES Hamiltonian. Assuming uniform nearest-neighbor coupling and choosing units such that $J/\hbar =1$ The controlled SES Hamiltonian is
\begin{equation*}
    H = \begin{pmatrix}
  \Delta_1      & 1 & \cdots & 0         & 1 \\
  1  & \Delta_2 & \cdots   & 0         & 0\\
  0        & 1 & \hdots     & 0         & 0\\
  \vdots   &   \vdots  & \ddots    & \vdots & \vdots \\
  1  & 0 &  \cdots  & 1 & \Delta_{N}
  \end{pmatrix}.
\end{equation*}
In terms of the Schr\"{o}dinger dynamics, the controlled state is the solution to the initial value problem
\begin{equation}
    \ket{\dot{\psi}(t)} = -iH \ket{\psi(t)}, \quad \ket{\psi(0)} = \ket{\psi_0}.
\end{equation}
The $\Delta_n$ are generated by maximizing the performance metric, defined as the fidelity, at a read-out time $t_f$ given by $F = |\braket{\psi_f|\psi(t_f)}|^2$. See~\cite{o'neil_2023_data_set_1,Schirmer_2017} for controller synthesis details. 
\subsection{Uncertainty Model}
To analyze the robustness of the controlled system, we use an uncertainty model that enables evaluation of how the introduction of uncertainty or disturbances alters the performance metric. In line with previous work~\cite{jonckheere_2018_jt_test,o'neil_2024_log_sens,o'neil_2023_cdc_2023} we consider real parametric uncertainty to the Hamiltonian represented as 
\begin{equation}\label{eq: perturbed_hamiltonian}
    \tilde{H}_{n} = H + \delta_n f_n S_n.
\end{equation}
$H$ is the nominal controlled Hamiltonian, and $\delta_n$ represents a small, real perturbation strength. $S_n$ is the structure associated with the uncertainty indexed by $n$.  $f_n$ is a factor that accounts for scaling by the control field strength for those $n$ that index uncertainty in a control channel; otherwise (e.g., for uncertainty in the couplings) $f_n = 1$. 
\subsection{Bloch Formulation}
To put the analysis in the framework of a real vector space, we use the Bloch formulation to describe the perturbed dynamics~\cite{floether_2012}. Consider the most general description of the Schr\"odinger dynamics where the density operator $\rho(t)$ for the pure state $\ket{\psi(t)}$ is given by $\ket{\psi(t)}\bra{\psi(t)}$. The perturbed dynamics generated by~\eqref{eq: perturbed_hamiltonian} are given by the von-Neumann equation
\begin{equation}\label{eq: von_neumann}
    \dot{\tilde{\rho}}(t) = -i [\tilde{H}_{n},\tilde{\rho}(t)], \quad \rho_0 = \ket{\psi_0}\bra{\psi_0}. 
\end{equation}
Taking the adjoint representation of this equation with respect to an orthonormal basis $\set{\sigma_n}_{n=1}^{N^2}$ of the $N \times N$ Hermitian matrices, such as the generalized Gell-Mann basis~\cite{bertlmann_2008_bloch}, yields the mapping $-i\tilde{H}_{n} \mapsto \mathrm{ad}_{-\imath \tilde{H}_n} := \tilde{\A} \in \mathbb{R}^{N^2 \times N^2}$ with elements 
\begin{equation}\label{eq: bloch_elements}
\tilde{\A}_{mn} = \tr\left( -\imath \tilde{H}_{n} [\sigma_m,\sigma_n] \right),
\end{equation}
where $[\cdot,\cdot]$ is the matrix commutator~\cite{floether_2012}. By linearity of the commutator and trace, $-iH \mapsto \A \in \mathbb{R}^{N^2 \times N^2}$ and $-iS_n \mapsto \Sn \in \mathbb{R}^{N^2 \times N^2}$ so that $\tilde{\A} = \A + \delta_n f_n \Sn$.  As $-iH$ and $-iS_n$ are skew-Hermitian, it follows from the expansion~\eqref{eq: bloch_elements} that $\A$ and $\Sn$ are skew-symmetric.  The state $\rho(t)$ maps to $\mathsf{r}(t) \in \mathbb{R}^{N^2}$ with elements $\mathsf{r}_m(t) = \tr \left( \rho(t) \sigma_m\right)$ so that $\rho_0 \mapsto \r0 \in \mathbb{R}^{N^2}$, $\ket{\psi_f}\bra{\psi_f}=: \rho_f  \mapsto \rf \in \mathbb{R}^{N^2}$, and the state equation is
\begin{equation}\label{eq: bloch_form}
    \dot{\tilde{\rr}}(t) = \tilde{\A}\tilde{\rr}(t), \quad \tilde{\rr}(0) = \r0
\end{equation}
 with solution $\tilde{\rr}(t) = e^{\tilde{\A}t}\r0$. The fidelity at the read-out time $t_f$ is $F = \rf^T e^{\tilde{\A}t_f} \r0$. The fidelity error is $\mathsf{e} = 1-F$. Finally, we define $\Phi:= e^{\mathsf{A}t_f}$ as the controlled propagator at $t_f$ and nominal $\A$ and define its perturbed counterpart as $\tilde{\Phi}$. 
\subsection{Differential Sensitivity}
The un-normalized differential sensitivity of $\mathsf{e}$ to a perturbation structured as $\Sn$ is given by $\zeta_n := \partial \mathsf{e} / \partial \delta_n$~\cite{o'neil_2024_log_sens,najfeld_1995_dexpma}.  Given the spectral decomposition $\A = \M \Lambda \M^\dagger$,
\begin{equation*}
\begin{split}
  \zeta_n & = \rf^T \left( \frac{\partial}{\partial \delta_n} e^{\tilde{\A}t_f}\right)\r0 = -t_f f_n u^{\dagger} \left( \Zn \odot \X \right)v \label{eq: K_n_in_eigenbasis}\\
  &= - t_f f_n \rf^T \M \left( \int_0^1 e^{t_f \Lambda (1-s)} \M^\dagger \Sn \M e^{t_f \Lambda s} ds \right) \M^{\dagger} \r0 
\end{split}
\end{equation*}
where $\Zn = \M^\dagger \Sn \M$, $u = \M^\dagger \rf$, $v = \M^\dagger \r0$, and $\odot$ denotes the Hadamard product. For matrices $A$ and $B$, the Hadamard product is defined such that the $(i,j)$ element of $A \odot B$ = $a_{ij}b_{ij}$. Denoting $\imath\lambda_k$ as the $k$th diagonal element of $\Lambda$, the entries of $\Zn \odot \X$ are 
\begin{equation}\label{eq: X_elements}
 \left(\Zn \odot \X \right)_{k \ell}= \begin{cases} z_{k k} e^{\imath\lambda_k t_f}, \quad & \lambda_k = \lambda_\ell \\ 
z_{k \ell}\frac{e^{\imath\lambda_k t_f} - e^{\imath\lambda_\ell t_f}}{it_f ( \lambda_k - \lambda_\ell) }, & \lambda_k \neq \lambda_\ell \end{cases}.
\end{equation}
Alternatively, we may define the operator $\mathsf{K}_n$ such that $\rf^T \left( \mathsf{K}_n \right) \r0 := u^\dagger \left( \Zn \odot \X \right)v$ and
\begin{equation}\label{eq: integral_Kn}
  \frac{-\zeta_n}{t_f f_n} =  \rf^T \mathsf{K}_n \r0 = \rf^T \left( \int_0^1 e^{t_f \mathsf{A}(1-s)} \Sn e^{t_f\mathsf{A}s} ds\right) \r0.
\end{equation}
\section{Geometric Interpretation of the Sensitivity}\label{sec: geometry_of_sensitivity}
\subsection{Fidelity versus Sensitivity Geometric Model}
Consider the space $\mathcal{B}(\mathbb{R}^{N^2})$ of bounded linear operators on $\mathbb{R}^{N^2}$~\cite{siewert_2022_matrix_bases}. In $\mathcal{B}(\mathbb{R}^{N^2})$, define $\R:=\rf \mathsf{r}_0^T$ as the \textit{input-output data operator}. $\Phi \in \mathcal{B}(\mathbb{R}^{N^2})$ may be called the \textit{input-output agnostic fidelity operator} as $F=\tr(\mathsf{R}^T\Phi)$. 
Likewise, we call $\Kn$ the \textit{input-output agnostic sensitivity operator} as justified by~\eqref{eq: integral_Kn}. 
Let $\mathcal{S}_n$ be the two-dimensional subspace of $\mathcal{B}(\mathbb{R}^{N^2})$ spanned by $\Phi$ and $\Kn$. 
The basic idea is to project the input-output operator $\mathsf{R}$ on $\mathcal{S}_n$ to introduce the input-output data in subspace of input-output agnostic operators. 
We use the Hilbert-Schmidt inner product so that for $A,B \in \mathcal{B}(\mathbb{R}^{N^2})$, $\braket{A,B} = \tr(A^T B)$ with compatible definition of the norm $\| A \|$. We define the orthogonal projection of $A$ on the span of $B$ as $\mathcal{P}_{B}(A) = \braket{A,\hat{B}}\hat{B}$ where $\hat{B} = B/\| B \|$.  
\begin{lemma}\label{lemmma: Phi_K_perp}
    The operators $\Phi$ and $\mathsf{K}_n$ are orthogonal. 
\end{lemma}
\begin{proof}
By definition of the inner product and the differential sensitivity, we have 
\begin{align*}
        \braket{\Phi,\mathsf{K}_n} & = \mathrm{Tr}\left( e^{-t_f \mathsf{A}} \left( \int_0^1 e^{t_f \mathsf{A} (1-s)} \Sn e^{t_f \mathsf{A} s} ds \right) \right) \\ &= \int_0^1 \mathrm{Tr} \left( e^{-t_f \mathsf{A} s} \Sn e^{t_f \mathsf{A} s} \right)ds = \mathrm{Tr} \left( \Sn \right).
\end{align*}
Now~\eqref{eq: bloch_elements} implies $\Sn \in \mathfrak{so}(N^2)$, being real and skew-symmetric, so $\mathrm{Tr}(\Sn) = 0$ for any allowed structure $S_n$ that retains unitary dynamics.
\end{proof}
\begin{theorem}
    The sensitivity and fidelity are related by the expression $|\zeta_n| = f_n t_f \| \Kn \| \cdot \| \Rs \| \sqrt{1 - \left( \frac{F}{N \| \Rs \| }\right)^2}$ where $\Rs$ is the projection of $\R:=\rf \mathsf{r}_0^T$ on the subspace $\mathcal{S}_n$.
\end{theorem}
\begin{proof}
By definition of $\R$, we have $F = \tr(\mathsf{R}^T \Phi) = \braket{\mathsf{R},\Phi}$ and $-\zeta_n / (f_n t_f) = \tr(\mathsf{R}^T \mathsf{K}_n) = \braket{\mathsf{R},\mathsf{K}_n}$. Then the projection of $\mathsf{R}$ on the subspace $\mathcal{S}_n$ is
\begin{equation}\label{eq: proj_S}
    \mathcal{P}_{\mathcal{S}}(\mathsf{R}) = \mathcal{P}_{\Phi}(\R) + \mathcal{P}_{\Kn}(\R) =: \mathsf{R}_{\mathcal{S}}. 
\end{equation}
This is equivalent to 
\begin{equation}\label{eq: Rs}
    \Rs = \frac{\braket{\R,\Phi}}{N^2} \Phi + \frac{\braket{\R,\Kn}}{\| \Kn \|^2} \Kn = \frac{F}{N^2}\Phi + \frac{-\zeta_n}{f_n t_f \| \Kn \|^2}\Kn,
\end{equation}
where $\| \Phi \| = N$ follows from $\Phi \in SO(N^2)$. Defining $\phi_n$ as the angle subtending $\Rs$ and $\Phi$ and $\theta_n$ as the angle subtending $\Kn$ and $\Rs$ yields
\begin{equation}\label{eq:trigonometry}
    \Rs = \left( \| \Rs \| \cos \phi_n \right) \hat{\Phi} + \left( \| \Rs \| \cos \theta_n \right) \hat{\mathsf{K}}_n. 
\end{equation}
\begin{figure}
\centering
{\includegraphics[width = 0.7\columnwidth]{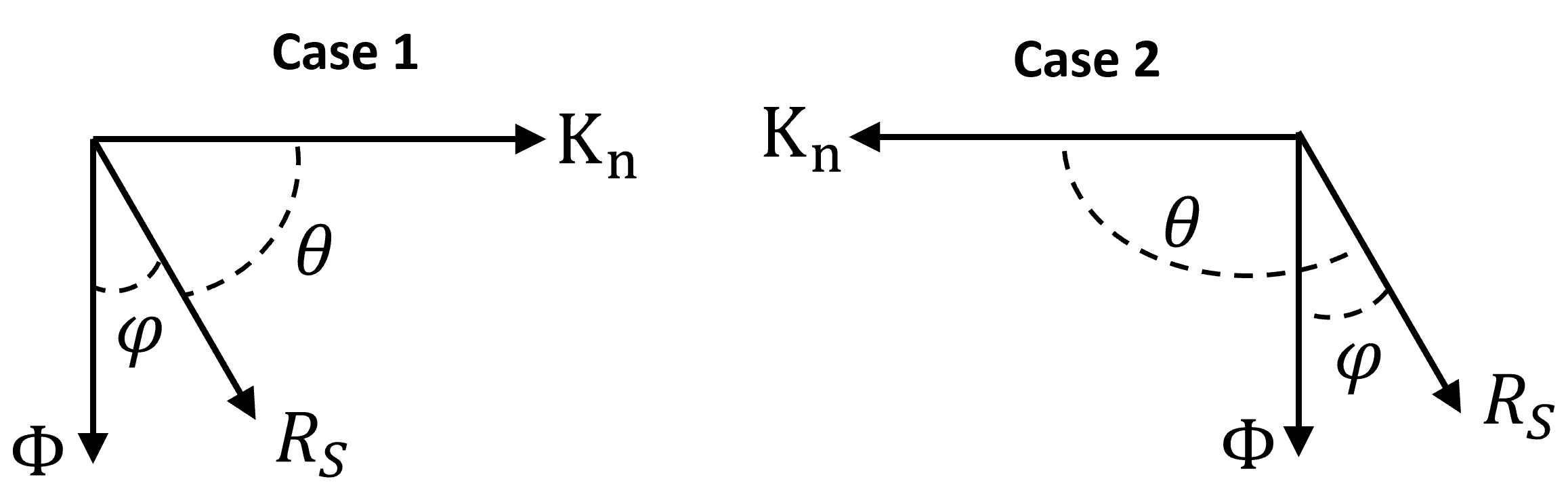}}
\caption{Relation of operators in the subspace $\mathcal{S}_n$.} 
\label{fig: geometry}
\end{figure}
Figure~\ref{fig: geometry} displays the relation of the operators in the subspace $\mathcal{S}_n$. It follows that $F = N \| \Rs \| \cos \phi_n$ and $ \zeta_n = -t_f f_n \| \Kn \| \cdot \| \Rs \| \cos \theta_n$.  By Lemma~\ref{lemmma: Phi_K_perp}, we have that $\Phi$ and $\Kn$ are orthogonal so that either $\phi_n+\theta_n = \pi/2 \mod \pi$ as in Case $1$ or $| \phi_n - \theta_n | = \pi/2$ as illustrated in Case $2$, so that $\cos \theta_n = \pm \sin \phi_n$. We then have following relation between the magnitude of the differential sensitivity and the fidelity
\begin{equation}\label{eq: sens_fidelity_1}
    |\zeta_n| = f_n t_f \| \Kn \| \cdot \| \Rs \| \cdot | \sin \phi_n |.
\end{equation}
Making the substitution $|\sin \phi_n| = \sqrt{1 - \left(  {F}/ \left( N \| \Rs \|  \right) \right)^2}$ completes the proof.
\end{proof}
Eqs.~\eqref{eq: proj_S}-\eqref{eq: Rs} appear to be the quantum equivalent of the classical relation $S+T=I$, 
expressed in terms of the fidelity $F$ rather than the error $\mathsf{e} = 1-F$. 
At the less conceptual level of~\eqref{eq: sens_fidelity_1}, 
holding $\|\Kn\|$ and $\|\Rs\|$ constant, $|\zeta_n|$ and $\mathsf{e}$ decrease simultaneously, in disagreement with the classical limitations mandating a trade-off between error and sensitivity. What undermines the simplistic formulation that  $|\zeta_n|$ and $\mathsf{e}$ are concordant 
is that $\|\Kn\|$ and $\|\Rs\|$ are not independent quantities; 
however, the first and crucial quantity can be bounded.  

\subsection{The Relation between $\texorpdfstring{\|\Kn\|}{Kn}$ and Eigenstructure}

To determine how $\| \Kn \|$ depends on the eigenstructure of the controlled propagator $\Phi$, we employ~\eqref{eq: X_elements} and $\Kn =\M \left( \Zn \odot \X \right) \M^\dagger =: \M \Qn \M^\dagger$.  It follows immediately that $\| \Kn \|^2 = \tr\left( \Qn^\dagger \Qn \right) = \sum_{k,\ell=1}^{N^2} |q_{k \ell}|^2$ where 
\begin{equation}\label{eq: norm_q_kell}
    |q_{k \ell}|^2 = \begin{cases} |z_{k \ell}|^2, &\lambda_k = \lambda_\ell  \\ |z_{k \ell}|^2 \sinc^2 \left( \frac{1}{2} \omega_{k \ell}t_f \right), & \text{otherwise} \end{cases}
\end{equation}
with the definition $\omega_{k \ell} := \left(\lambda_k - \lambda_\ell \right)$.  We now identify the maximum of $\| \Kn \|$ and establish strict positivity. 
\begin{lemma}\label{lemma: norm_Kn}
    $\| \Kn \|$ is non-zero and bounded above by $\| \Sn \|$. 
\end{lemma}
\begin{proof}
To see that $\| \Kn \|$ is not zero it suffices to show that $\Kn^T \Kn$ is never the zero matrix.  Consider the integral expression of $\Kn$ following from~\eqref{eq: integral_Kn}. Pulling $e^{\A t_f} = \Phi$ out of $\Kn^T$ and $\Kn$ and using $\Phi^T \Phi = I$, we are left with $\Kn^T \Kn = \left(\int_0^1 e^{-t_f \mathsf{A} s} \Sn e^{t_f \mathsf{A} s} ds\right)^T \left(\int_0^1 e^{-t_f \mathsf{A} \tau} \Sn e^{t_f \mathsf{A} \tau} d \tau \right)$.  Noting that $e^{\pm t_f \A  s} \in SO(N^2)$ and $\Sn \in \mathfrak{so}(N^2)$, shows that the integrand is the adjoint action of the special orthogonal group on its Lie algebra for any $s$~\cite{elliott_2009_matrices_in_action}. This Lie group conjugation is described by the mapping $\mathrm{Ad}_{(e^{-t_f \A  s})}: \Sn \mapsto e^{-t_f \A  s} \Sn e^{t_f \A s}$.  If the kernel of the integral of $\mathrm{Ad}_{(e^{-t_f \A s})}$ is trivial, it follows that $\Kn^T \Kn$ is not the zero matrix for non-trivial $\mathsf{S}_n$. The spectrum of $\A$ given by $\set{\imath \lambda_k}_{k=1}^{N^2}$ determines the eigenvalues of $\mathrm{Ad}_{(e^{-t_f \A s})}$, which are $e^{-\imath t_f (\lambda_k - \lambda_\ell) s}$ for $k,\ell$ indexed from $1$ to $N^2$.  The eigenvalues of $\mathrm{Ad}_{(e^{-t_f \A s})}$ retain the dependence of the integration variable $s$ (equivalently $\tau$), as can be verified by an explicit spectral decomposition of $\mathrm{Ad}_{(e^{-t_f \A s})}$.  Integrating $s$ over the interval $[0,1]$ yields eigenvalues for $\int_0^1 \mathrm{Ad}_{(e^{-t_f \A s})} ds$ of $1$ for $\lambda_\ell = \lambda_k$ and $\imath \left( e^{-\imath (\lambda_k - \lambda_\ell)t_f} - 1 \right) /\left( t_f (\lambda_k - \lambda_\ell) \right)$ otherwise.  As these eigenvalues are never zero, the kernel of the integral of $\mathrm{Ad}_{(e^{-t_f \A s})}$ is trivial.  It follows that, for a non-trivial uncertainty structure $\Sn$, $\| \Kn \|$ is non-zero. To establish the upper bound, note that when all $\sinc$ terms in~\eqref{eq: norm_q_kell} assume the maximum of $1$, then $\| \Kn \| \leq \sqrt{\sum_{k,\ell=1}^{N^2} |z_{k \ell}|^2} = \sqrt{\tr \left(\Zn^\dagger \Zn \right)} =  \sqrt{ \tr\left(\M^\dagger \mathsf{S}_n^\dagger \M \M^\dagger \mathsf{S}_n \M  \right)} = \| \Sn \|$. 
\end{proof}
Replacing $\|\Kn\|$ by its upper bound in~\eqref{eq: sens_fidelity_1} yields an upper bound on $|\zeta_n|$, which decreases as $\mathsf{e}$ decreases. \textit{But this behavior on the bound does not rule out an increase in $|\zeta_n|$ as $\mathsf{e}$ decreases}~\cite{o'neil_2024_log_sens}. This indicates that any quantum performance limitation is not as straightforward as in classical control and that $|\zeta_n|$ and $\mathsf{e}$ could be concordant~\cite{jonckheere_2018_jt_test,o'neil_2023_data_set_1} or discordant~\cite{o'neil_2024_log_sens}.
\subsection{Sufficient Condition for Vanishing Sensitivity}
\begin{lemma}\label{lemma: zero_sensitivity}
For unitary evolution, the differential sensitivity $\zeta_n$ vanishes for any physically allowable perturbation if the controlled propagator $\Phi$ induces perfect state transfer.
\end{lemma}
\begin{proof}
As a physically realizable perturbation structure, $S_n$ is Hermitian, and its Bloch representation $\Sn$ is skew-symmetric.  For non-zero $t_f$ and $f_n$, $-\zeta_n/(f_n t_f)$ is given by 
    \begin{align*}
        &\braket{\R,\Kn} = \mathrm{Tr}\left( \R^T \left( \int_0^1 e^{t_f \mathsf{A} (1-s)} \Sn e^{t_f \mathsf{A} s} ds \right) \right) \\ &= \mathrm{Tr} \left( \R^T \Phi  \left( \int_0^1 e^{-t_f \mathsf{A} s} \Sn e^{t_f \mathsf{A} s}ds \right) \right) = \tr \left( \R^T \Phi \W \right). 
    \end{align*}  
From Lemma~\ref{lemma: norm_Kn}, the term in the integral is the adjoint action of $SO(N^2)$ on $\mathfrak{so}(N^2)$ for any value of $s$.  Thus $\W \in \mathfrak{so}(N^2)$, and $\braket{\R,\Kn} = 0$ is secured if the product $\R^T \Phi =\mathsf{r}_0 \rf^T \Phi $ is symmetric. We thus require $\r0 \left( \rf^T \Phi \right) = (\Phi^T \rf) \r0^T$.  Symmetry of this dyadic product requires $\r0=\lambda \Phi^T\rf$ for some real $\lambda$. Since $\|\rf\|=\|\r0\|=1$ and $\Phi$ is orthogonal, it follows that $\lambda=\pm 1$. For $\lambda=1$, $\rf=\Phi \r0$, which is the condition for perfect state transfer from input state $\r0$ to $\rf$.  Thus, for $\lambda=1$, a controller inducing perfect state transfer is sufficient for $| \zeta_n | = 0$. For $\lambda=-1$, we would have $\rf=-\Phi\r0$. Observing that $(\rf)_m=\mathrm{Tr}(\rho_{f} \sigma_m)$ with $\rho_0 \geq 0$, it follows that $(-\rf)_m=\mathrm{Tr}(-\rho_{f} \sigma_m)$, and $(-\rf)$ would be associated with a nonpositive definite density, which is absurd. 
\end{proof}
\subsection{Properties of $ \Rs $}\label{subsection: properties of Rs}
\begin{remark}\label{obs: upper_Rs}
    At perfect state transfer $\| \Rs \| = 1/N$ and $\cos \phi_n = 1$.  This follows from~\eqref{eq: Rs} and the orthogonality of $\hat{\Phi}$ and $\hat{\mathsf{K}}_n$, which implies $\| \Rs \|^2 = (F/N)^2 + (\zeta_n/(f_n t_f \| \Kn \|))^2$.  Lemma~\ref{lemma: zero_sensitivity} shows that $\zeta_n = 0$ and thus $\| \Rs \| = (1/N)$ for $F=1$. The relation $F = N \| \Rs \| \cos \phi_n$ further implies that $\cos \phi_n = 1$ for perfect state transfer.
\end{remark}
\begin{remark}
    For any fidelity, $\|\Rs\|$ is bounded below by $\| \Rs \| \geq F/N$. This follows directly from~\eqref{eq: sens_fidelity_1}, which requires $F/\left(N \| \Rs\| \right) < 1$ for the sensitivity $\zeta_n$ to be real.
\end{remark}
Though Remark~\ref{obs: upper_Rs} only provides the upper bound on $\| \Rs \|$ for perfect state transfer, we find that this holds as an upper bound in general.  Figure~\ref{fig: N=5_Rs} for the case of a $5$-ring shows that $\| \Rs \|$ remains below $1/N$ for all controllers.
\begin{figure}
\centering
{\includegraphics[width = 0.9\columnwidth]{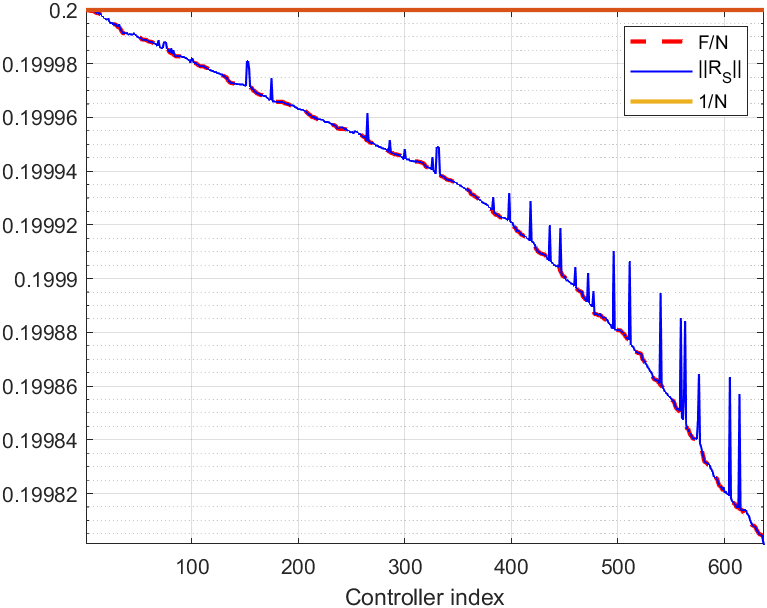}}
\caption{Plot of $\|\Rs\|$ and $1/F$ for a $5$-ring and $\ket{1} \rightarrow \ket{3}$ transfer with coupling uncertainty indexed by $n=10$.} 
\label{fig: N=5_Rs}
\end{figure}
\section{Vanishing Differential Sensitivity}
We now extend the result of~\cite[Th. 3]{Schirmer_2017} for sufficient conditions on vanishing sensitivity in spintronic networks.
\begin{theorem}\label{thm: vanishing_sensitivity}
For unitary evolution and controllers that yield non-zero fidelity, the differential sensitivity $\zeta_n$ vanishes for any physically allowable perturbation or uncertainty structure if and only if the controlled propagator induces perfect state transfer. 
\end{theorem}
\begin{proof}
Sufficiency is established by Lemma~\ref{lemma: zero_sensitivity}. For necessity, we refer to~\eqref{eq: sens_fidelity_1}.
For any physically realizable controller, both $f_n$ and $t_f$ are non-zero.
Moreover, $\| \Kn \|\ne 0$ by Lemma~\ref{lemma: norm_Kn}.
Hence, for $|\zeta_n|=0$, either $\|\Rs\|=0$, $\sin \phi_n=0$, or both. But if $\Rs=0$, it follows from Eqs.~\eqref{eq: proj_S}-\eqref{eq:trigonometry} that $F=0$ and $\zeta_n=0$. 
Since $\zeta_n=0$ is already assumed, the only consequence is $F=0$, which is disallowed by hypothesis.
So $\| \Rs \| \neq 0$ unless the fidelity vanishes. Then $| \zeta_n | = 0$ requires $\sin \phi_n = 0$. But this is precisely the configuration of the operators in $\mathcal{S}_n$ for perfect state transfer.   
\end{proof}
\section{Examples}\label{sec: examples}
To demonstrate the utility of this model, we examine the fidelity error and sensitivity data for two state-transfer examples from the data in~\cite{o'neil_2023_data_set_1}. We choose these cases as exemplary of the varying sensitivity-fidelity trends seen when the controller does not induce perfect state transfer as in Theorem~\ref{thm: vanishing_sensitivity}.
\subsection{Differing Fidelity Error-Sensitivity Profiles}
\begin{figure}
\centering
{\includegraphics[width = 0.9\columnwidth]{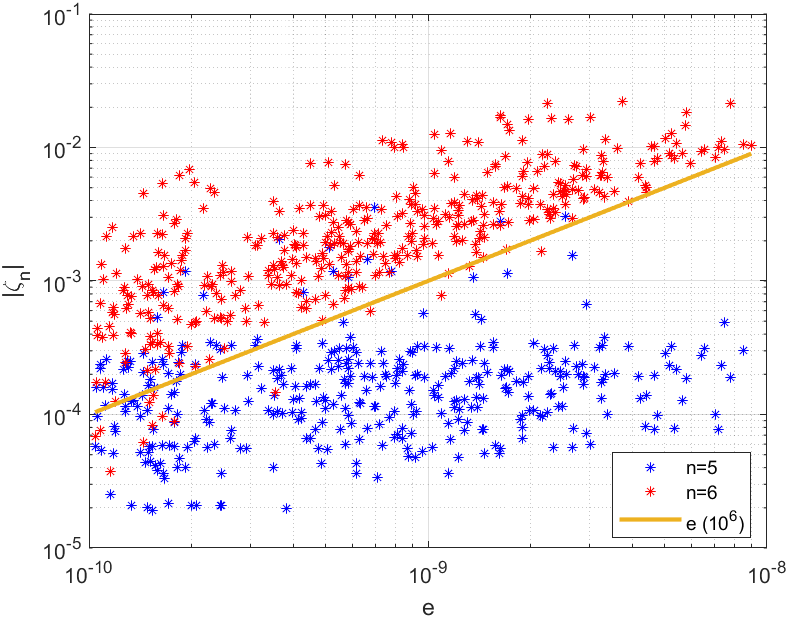}}
\caption{Plot of $|\zeta_n|$ versus $\mathsf{e}$ for an $N=4$, $\ket{1} \rightarrow \ket{2}$ transfer. The solid line indicates a line of unity slope on a log-log scale.} 
\label{fig: N=4_scatter}
\end{figure}
We begin with the case of state transfer from spin $1$ to spin $2$ in a $4$-ring with control mediated by time-invariant controls.  We index the possible perturbations to the controls, the diagonal elements of the SES Hamiltonian, by $n=1$ though $n=4$. We index uncertainty in the couplings (i.e. uncertainty in the entries $((n-4),(n-3))$ and $((n-3),(n-4)))$ of the SES Hamiltonian) by $n=5$ though $n=7$. Uncertainty to the entries $(1,4)$ and $(4,1)$ is indexed by $n=8$. For all cases aside from $n=5$ and $n=7$, $\log{\mathsf{e}}$ and $\log|\zeta_n|$ display a strong linear correlation with a Pearson $r$ greater than $0.75$. For the remaining two cases the correlation coefficient is less than $0.14$. To investigate the origin of this differing behavior, we examine the $n=5$ and $n=6$ uncertainty cases as depicted in Figure~\ref{fig: N=4_scatter}. The strong linear relationship for the $n=6$ case is borne-out by a Pearson $r$ of $0.7503$, while the much flatter trend for the $n=5$ case can be verified by the Pearson $r$ of $0.1373$. 
\begin{figure}
\centering
{\includegraphics[width = 0.9\columnwidth]{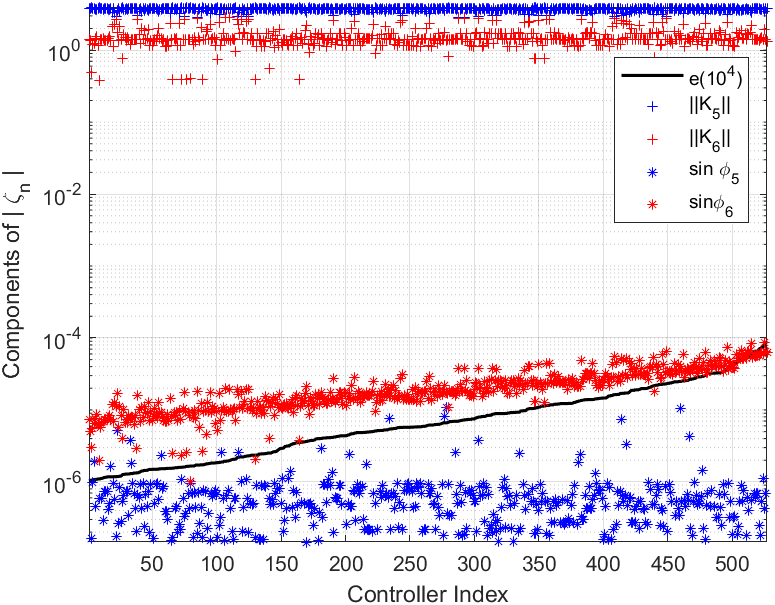}}
\caption{Plot of components of $|\zeta_n|$ from Figure~\ref{fig: N=4_scatter}. While the variation in $\| \Kn \|$ is minimal, $\sin \phi_n$ trends positively with the fidelity error for $n=6$ but shows a flat trend for $n=5$.} 
\label{fig: N=4_components}
\end{figure}
We explain these different trends with the geometric model of Section~\ref{sec: geometry_of_sensitivity} and~\eqref{eq: sens_fidelity_1}.  Consider the visualization of the components of $|\zeta_n|/t_f$ shown in Figure~\ref{fig: N=4_components}.  Note that both cases correspond to coupling uncertainty, so $f_n=1$.  For both cases, $\| \Rs \|$ is nearly constant at $0.25$ across all controllers with a maximum deviation below this value on the order of $10^{-10}$.  This data is not displayed in Figure~\ref{fig: N=4_components}. $\| \Kn \|$ shows some deviation across the controllers, however these deviations are unlikely to induce the different trends observed in Figure~\ref{fig: N=4_scatter}.  In particular for $n=5$, $\| \Kn \|$ has a mean of $3.650$ and variance of $0.0759$, while for $n=6$ the statistics for $\| \Kn \|$ are $1.527$ and $0.2752$. However, as seen from Figure~\ref{fig: N=4_components}, the behavior of $\sin \phi_n$ is markedly different for the controllers in the $n=5$ versus $n=6$ cases. The positive rank correlation between $\sin \phi_6$ and $\mathsf{e}$ is borne out by a Kendall $\tau$ of $0.8075$. Conversely, the lack of rank correlation between metrics for the $n=5$ case is evident in the weak Kendall $\tau$ of $-0.0275$. 
\par This difference in behavior has the geometric interpretation that for the $n=5$ structure, any change in fidelity (error) is more attributable to the increase in the angle defining the inclination of $\R$ with the subspace $\mathcal{S}_5$ than to an increase of $\phi_5$ within the subspace. This minimal change of $\phi_5$ with the fidelity error manifests as the absence of a clear trend between $|\zeta_5|$ and $\mathsf{e}$ as the fidelity changes. Conversely, for the $n=6$ case, an increase in the fidelity error is strongly correlated with, and nearly proportional to, an increase of the angle $\phi_6$ within the subspace $\mathcal{S}_6$. This suggests that those uncertainties that generate a subspace $\mathcal{S}_n$ spanned by $\Phi$ and $\Kn$ where a change in the fidelity corresponds to a change in $\phi_n$, are those that exhibit the greatest sensitivity. This bolsters the assertion that the fidelity-sensitivity relation is not a simple trade-off and must account for geometric factors that may not be directly physically accessible.  
\subsection{Large Variation in Sensitivity for Nearly Equal Error}
We now study a case where there is no trend between $|\zeta_n|$ and $\mathsf{e}$ across controllers for the same uncertain parameter. We consider transfer from spin $1$ to spin $2$ in a $5$-ring with perturbation structure $n=5$ (perturbation to the control field addressing spin $5$). 
\begin{figure}
\centering
{\includegraphics[width = 0.9\columnwidth]{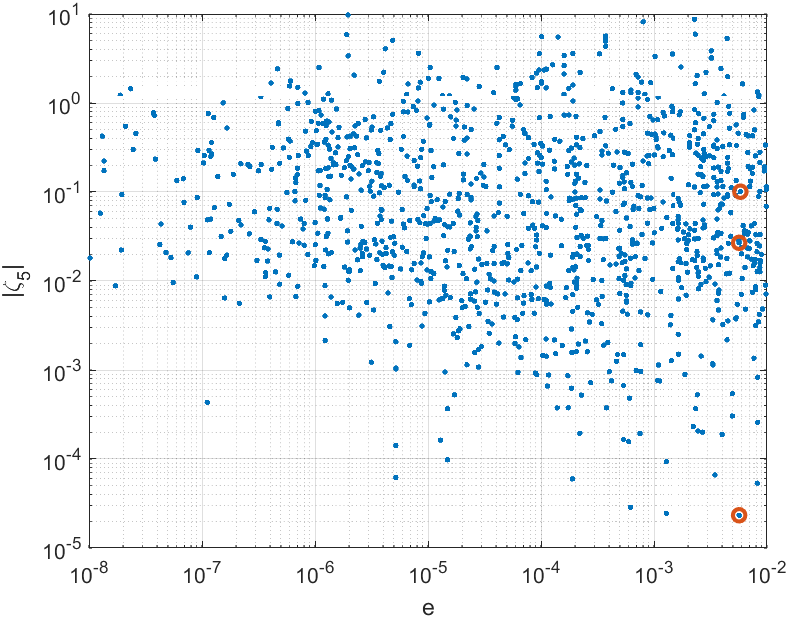}}
\caption{Plot of $|\zeta_5|$ versus $\mathsf{e}$ for an $N=5$, $\ket{1} \rightarrow \ket{2}$ transfer. The data points circled in red display widely varying sensitivity for nearly equal error.} 
\label{fig: N=5_scatter}
\end{figure}
\setlength{\tabcolsep}{5pt}
\renewcommand{\arraystretch}{0.2} 
\begin{table}[b]
\caption[Summary of Parameters Contributing to Sensitivity]{Physical and geometric factors contributing to $| \zeta_5 |$}\label{table: geometry}
\centering
\begin{tabular}{|c|c|c|c|c|c|c|}
\hline
Index & $f_5$ & $t_f$ & $\| \mathsf{K}_5 \|$& $\| \Rs \|$ & $\sin \phi_5$ & $| \zeta_5 |$ \\ \hline
$1585$ & $2.35$ & $198$ & $1.71$ & $0.199$ & $1.48 \times 10^{-7}$ & $2.33 \times 10^{-5}$                           \\ \hline
$1586$ & $330$  & $108$ & $2.74$ & $0.199$ & $1.37 \times 10^{-6}$ &  $2.69 \times 10^{-2}$                          \\ \hline
$1587$ & $17.0$ & $400$  & $1.68$ & $0.198$ & $4.50 \times 10^{-5}$                          & $1.01 \times 10^{-1}$                           \\ \hline
\end{tabular}
\end{table}
Figure~\ref{fig: N=5_scatter} depicts the sensitivity versus error profile, which does not reveal any visual trend between the two metrics. Rather, we see that the controllers with the lowest absolute value of sensitivity fall in the range of $10^{-4} < \mathsf{e} < 10^{-2}$. We focus on the controller indices $1585$, $1586$, and $1587$ (the red circled data points in Figure~\ref{fig: N=5_scatter}) which have an error in the range $5.67 \times 10^{-3} \pm 5 \times 10^{-5}$, while the differential sensitivity for these controllers spans orders of magnitude from $2.33 \times 10^{-5}$ for controller $1585$ up to $0.101$ for controller $1587$. We provide insight into this vast change in sensitivity by considering the geometric factors encoded in the size and orientation of the related matrix operators and the physical parameters $t_f$ and $f_n$ as captured in~\eqref{eq: sens_fidelity_1}. Table~\ref{table: geometry} provides a summary of the relevant factors contributing to $| \zeta_5 |$ for each controller. From this, we see that unlike in the previous example, the behavior cannot be attributed simply to the effect of $\sin \phi_5$, but rather to a combination of the factors $\sin \phi_5$, $t_f$, and $f_5$. The extremely small sensitivity of controller $1585$ is most attributable to a small angle $\phi_5$ between $\Phi$ and $\Rs$ and small control amplitude.  While controller $1586$ demonstrates a $\sin \phi_5$ only an order of magnitude larger than the previous controller and a shorter read-out time, the large value of the control field contributes to an increase of over two orders of magnitude in $| \zeta_5 |$.  Finally, while controller $1587$ admits a smaller control amplitude than its predecessor, the larger read-out time and $\sin \phi_5$, manifest as the largest sensitivity of the trio. This brief analysis justifies the premium placed on minimizing transfer times to increase robustness beyond the desire to minimize the impact of decoherence in open systems in general~\cite{koch_2016_open_systems}.  It also supports efforts to limit control amplitudes beyond energy considerations.  Finally, optimizing for small $\phi_n$ in combination with minimizing the error has the potential to produce controllers with reduced sensitivity in comparison to those optimized simply for minimum error. 
\section{Conclusion}\label{sec: conclusion}
We developed a geometric model of the sensitivity to parametric uncertainty for state transfer problems that analytically relates sensitivity and performance. We employed this model to provide insight into results of previous work based on statistical analysis of the sensitivity versus fidelity error~\cite{jonckheere_2018_jt_test}. With this model we expanded the scope of another previous work~\cite[Th. 3]{Schirmer_2017} by proving that perfect fidelity is not only sufficient, but necessary, for vanishing sensitivity. Future work should focus on relating the geometric parameters $\Rs$ and $\phi_n$ to the eigenstructure of the controlled propagator. Such an approach would allow for direct synthesis of minimum-sensitivity controllers by penalizing the impact of these parameters during controller optimization. Future work should also consider the adaptation of this model to piecewise constant and generally continuous control pulses as well as to applications beyond state transfer such as gate operation.  

\bibliographystyle{ieeetr}
\bibliography{IEEEabrv,bibliography}

\begin{thebibliography}{10}

\bibitem{koch_2022_quantum_optimal_control_survey}
{Christiane P. Koch \emph{et al.}}, ``Quantum optimal control in quantum technologies. {S}trategic report on current status, visions and goals for research in {E}urope,'' {\em {EPJ} Quantum Technology}, vol.~9, 7 2022.

\bibitem{Glaser_2015}
{Steffen J. Glaser \emph{et al.}}, ``Training {S}chrödinger's cat: quantum optimal control,'' {\em {EJP D}}, vol.~69, 12 2015.

\bibitem{wang_2023}
S.~Wang, C.~Ding, Q.~Fang, and Y.~Wang, ``Quantum robust optimal control for linear complex quantum systems with uncertainties,'' {\em {IEEE} Trans. Autom. Control}, vol.~68, no.~11, pp.~6967--6974, 2023.

\bibitem{James_2007}
M.~R. James, H.~I. Nurdin, and I.~R. Petersen, ``{$H^{\infty}$} control of linear quantum stochastic systems,'' {\em {IEEE} Trans. Autom. Control}, vol.~53, no.~8, pp.~1787--1803, 2008.

\bibitem{khalid_2023_rim}
I.~Khalid, C.~A. Weidner, E.~A. Jonckheere, S.~G. Shermer, and F.~C. Langbein, ``Statistically characterizing robustness and fidelity of quantum controls and quantum control algorithms,'' {\em Phys. Rev. A}, vol.~107, Mar. 2023.

\bibitem{koswara_2021}
A.~Koswara, V.~Bhutoria, and R.~Chakrabarti, ``Robust control of quantum dynamics under input and parameter uncertainty,'' {\em Phys. Rev. A}, vol.~104, p.~053118, Nov 2021.

\bibitem{Kosut_2013}
R.~L. Kosut, M.~D. Grace, and C.~Brif, ``Robust control of quantum gates via sequential convex programming,'' {\em Phys. Rev. A}, vol.~88, nov 2013.

\bibitem{jonckheere_2018_jt_test}
E.~Jonckheere, S.~Schirmer, and F.~Langbein, ``Jonckheere‐{T}erpstra test for nonclassical error versus log‐sensitivity relationship of quantum spin network controllers,'' {\em Int. J. Robust Nonlinear Control.}, vol.~28, p.~2383–2403, jan 2018.

\bibitem{o'neil_2024_log_sens}
S.~O'Neil, S.~Schirmer, F.~C. Langbein, C.~A. Weidner, and E.~A. Jonckheere, ``Time-domain sensitivity of the tracking error,'' {\em IEEE Transactions on Automatic Control}, vol.~69, no.~4, pp.~2340--2351, 2024.

\bibitem{o'neil_2023_data_set_1}
S.~P. O’Neil, F.~C. Langbein, E.~Jonckheere, and S.~Shermer, ``Robustness of energy landscape controllers for spin rings under coherent excitation transport,'' {\em Research Directions: Quantum Technologies}, vol.~1, p.~e12, 2023.

\bibitem{Jonckheere_2017}
E.~A. Jonckheere, S.~G. Schirmer, and F.~C. Langbein, ``Structured singular value analysis for spintronics network information transfer control,'' {\em {IEEE} Trans. Autom. Control}, vol.~62, no.~12, pp.~6568--6574, 2017.

\bibitem{donovan_2011_hamitonian_manipulation}
A.~Donovan, V.~Beltrani, and H.~A. Rabitz, ``Quantum control by means of hamiltonian structure manipulation.,'' {\em Phys. Chem. Chem. Phys.}, vol.~13 16, pp.~7348--62, 2011.

\bibitem{zhang_2023_gate_field_control}
R.-W. Zhang, C.~Cui, R.~Li, J.~Duan, L.~Li, Z.-M. Yu, and Y.~Yao, ``Predictable gate-field control of spin in altermagnets with spin-layer coupling,'' {\em Phys. Rev. Lett.}, vol.~133, p.~056401, Aug 2024.

\bibitem{Schirmer_2017}
S.~G. Schirmer, E.~A. Jonckheere, and F.~C. Langbein, ``Design of feedback control laws for information transfer in spintronics networks,'' {\em {IEEE} Trans. Autom. Control}, vol.~63, no.~8, pp.~2523--2536, 2018.

\bibitem{o'neil_2023_cdc_2023}
S.~P. O’Neil, I.~Khalid, A.~A. Rompokos, C.~A. Weidner, F.~C. Langbein, S.~Schirmer, and E.~A. Jonckheere, ``Analyzing and unifying robustness measures for excitation transfer control in spin networks,'' {\em IEEE Control Syst. Lett.}, vol.~7, pp.~1783--1788, 2023.

\bibitem{floether_2012}
F.~F. Floether, P.~de~Fouquieres, and S.~G. Schirmer, ``Robust quantum gates for open systems via optimal control: {M}arkovian versus non-{M}arkovian dynamics,'' {\em New J. Phys.}, vol.~14, p.~073023, jul 2012.

\bibitem{bertlmann_2008_bloch}
R.~A. Bertlmann and P.~Krammer, ``Bloch vectors for qudits,'' {\em J. Phys. A-Math.}, vol.~41, p.~235303, may 2008.

\bibitem{najfeld_1995_dexpma}
I.~Najfeld and T.~Havel, ``Derivatives of the matrix exponential and their computation,'' {\em Adv. Appl. Math.}, vol.~16, no.~3, pp.~321--375, 1995.

\bibitem{siewert_2022_matrix_bases}
J.~Siewert, ``On orthogonal bases in the {H}ilbert-{S}chmidt space of matrices,'' {\em J. Phys. Commun.}, vol.~6, p.~055014, May 2022.

\bibitem{elliott_2009_matrices_in_action}
D.~Elliott, {\em Bilinear Control Systems: Matrices in Action}.
\newblock Springer Publishing Company, Inc., 1st~ed., 2009.

\bibitem{koch_2016_open_systems}
C.~P. Koch, ``Controlling open quantum systems: tools, achievements, and limitations,'' {\em J. Condens. Matter Phys.}, vol.~28, p.~213001, 2016.

\end{thebibliography}

\end{document}